\documentclass{endm}
\usepackage{endmmacro}
\usepackage{graphicx,amsmath}

\begin{document}

\begin{verbatim}\end{verbatim}\vspace{2.5cm}

\begin{frontmatter}

\title{Proper Hamiltonian Paths in Edge-Coloured Multigraphs}

\author[ESP]{Raquel \'Agueda}
\author[LRI]{Valentin Borozan}
\author[DC]{Marina Groshaus\thanksref{M}}
\author[LRI]{Yannis Manoussakis} 
\author[LRI]{Gervais Mendy}
\author[LRI]{Leandro Montero}

\address[ESP]{Departamento de An\'alisis Econ\'omico y Finanzas, Universidad de Castilla, \\ 
La Mancha 45071 Toledo, Spain. \\{\rm \texttt{raquel.agueda@uclm.es}}}

\address[LRI]{L.R.I.,  B\^at. 650, Universit\'e Paris-Sud 11, \\ 
91405 Orsay Cedex, France. \\{\rm \texttt{\{valik,yannis,mendy,lmontero\}@lri.fr}}}

\address[DC]{Departamento de Computaci\'on, FCEyN, Universidad de Buenos Aires, \\ 
Buenos Aires, Argentina. \\{\rm \texttt{groshaus@dc.uba.ar}}}

\thanks[M]{Partially supported by UBACyT X456, X143 and ANPCyT PICT 1562 Grants and by \mbox{CONICET}, Argentina.}

\begin{abstract} 
Given a $c$-edge-coloured multigraph, a proper Hamiltonian path is a path that contains all the vertices 
of the multigraph such that no two adjacent edges have the same colour. In this work we establish sufficient conditions for an 
edge-coloured multigraph to guarantee the existence of a proper Hamiltonian path, involving various parameters as the number of edges,
the number of colours, the rainbow degree and the connectivity.
\end{abstract}


\end{frontmatter}

\section{Introduction}\label{intro}
The study of problems modelled by edge-coloured graphs have resulted in important
developments recently. For instance, the research on long coloured cycles
and paths for edge-coloured graphs has provided interesting results~\cite{Bang-Jensen2001}. 
From a practical perspective, problems arising in molecular biology are often modeled using coloured
graphs, i.e., graphs with coloured edges and/or vertices~\cite{Pevzner2000}. Given an edge-coloured
graph, the original problems are equivalent to extract subgraphs coloured in a specified pattern.
The most natural pattern in such a context is that of proper colourings, i.e., adjacent edges
have different colours. 

In this work we give sufficient conditions involving various parameters as the number of edges, rainbow degree, etc,
in order to guarantee the existence of proper Hamiltonian paths in edge-coloured multigraphs where parallel edges with same colours 
are not allowed. Notice that the proper Hamiltonian path and proper Hamiltonian cycle problems are both $NP$-complete in the general case.
However it is polynomial to find a proper Hamiltonian path in $c$-edge-coloured complete graphs, $c \geq 2$~\cite{Feng}.
It is also polynomial to find a proper Hamiltonian cycle in $2$-edge-coloured complete graphs~\cite{bankfalvi},
but it is still open to determine the computational complexity for $c\geq 3$~\cite{Benkouar}.
Many other results for edge-coloured multigraphs can be found in the survey by Bang-Jensen and Gutin~\cite{BG97}.
Results involving only degree conditions can be found in~\cite{Abouelaoualim2010}.

Formally, let $I_c=\{1,2,\ldots, c\}$ be a set of $c \geq 2$ colours. Throughout this paper, $G^c$ denotes a \emph{c-edge-coloured 
multigraph} such that each edge is coloured with one colour in $I_c$ and no two parallel edges joining the same pair of vertices 
have the same colour. Let $n$ be the number of vertices and $m$ be the number of edges of $G^c$.
If $H$ is a subgraph of $G^c$, then $N^i_H(x)$ denotes
the set of vertices of $H$ adjacent to $x$ with an edge of colour $i$. Whenever $H$ is isomorphic to $G^c$, 
we write $N^i(x)$ instead of $N^i_{G^c} (x)$. The \emph{coloured i-degree} of a vertex $x$,
denoted by $d^i(x)$, is the cardinality of $N^i(x)$. As usual $N(x)$ denotes the neighbourhood of $x$, $d(x)$ its degree and 
$\delta(G)$ the minimum degree among all vertices of $G^c$.
The \emph{rainbow degree} of a vertex $x$, denoted by $rd(x)$, 
is the number of different colours on the edges incident to $x$. The \emph{rainbow degree} of a multigraph $G^c$, denoted by $rd(G^c)$, is 
the minimum rainbow degree among its vertices. An edge with endpoints $x$ and $y$ is denoted by 
$xy$, and its colour by $c(xy)$. A \emph{rainbow complete multigraph} is the one having all possible coloured edges between any pair of vertices 
(its number of edges is therefore $c\binom{n}{2}$). The \emph{complement} of a multigraph $G^c$ denoted by $\overline{G^c}$, is 
a multigraph with the same vertices as $G^c$ and an edge $vw \in E(\overline{G^c})$ on colour $i$ if and only if $vw \notin E(G^c)$ on that colour.
We say that an edge $xy$ is a \emph{missing edge} of $G^c$ if $xy \in E(\overline{G^c})$.
The graph $G^i$ is the spanning subgraph of $G^c$ with edges only in colour $i$.
A subgraph of $G^c$ is said to be \emph{properly edge-coloured} if any two adjacent edges in this subgraph differ 
in colour. A \emph{Hamiltonian path} (\emph{cycle}) is a path (cycle) containing all vertices of the multigraph. A path is said to be 
\emph{compatible} with a given matching $M$ if the edges of the path are alternatively in $M$ and not in $M$. We assume that the first and the last edge of the path 
are in $M$ otherwise we just remove one (or both) of them in order to have this property.
All multigraphs are assumed to be connected.

This paper is organized as follows: In Section~\ref{prem_res} we present some 
preliminary results that will be useful for the rest of the paper. In Section~\ref{2_edge_col} 
we study proper Hamiltonian paths in $2$-edge-coloured multigraphs. In Section~\ref{3_edge_col} 
we study proper Hamiltonian paths in $c$-edge-coloured multigraphs, for $c \geq 3$. 
Notice that in this work we focus only on edge-coloured multigraphs since it makes no sense to study such conditions for simple edge-coloured graphs.

\section{Preliminary results}\label{prem_res}

\begin{lemma}\label{matchingsPerfect} Let $G$ be a connected non-coloured simple
graph on $n$ vertices, $n \geq 9$. If $m\geq \binom{n-2}{2} + 3$, then
$G$ has a matching $M$ of size $|M|= \lfloor \frac{n}{2} \rfloor$.  
\end{lemma}
\begin{proof} By a theorem in~\cite{ByerSmeltzerDM2007}, a $2$-connected graph
on $n \geq 10$ vertices and $m\geq \binom{n-2}{2} + 5$ edges has a
Hamiltonian cycle. So if we add a new vertex $v$ to $G$ and we join it to all the 
vertices of $G$ we have that $G+\{v\}$ has $m\geq \binom{n-1}{2} + 5$ edges.  Therefore $G+\{v\}$ has a Hamiltonian
cycle, i.e., $G$ has a Hamiltonian path and this implies that there exists a
matching $M$ in $G$ of size $|M|= \lfloor \frac{n}{2} \rfloor$.  
\end{proof}

\begin{lemma}[\cite{leathesis}]\label{matching} Let $G$ be a simple non-coloured graph on $n\geq
14$ vertices. If $m\geq \binom{n-3}{2}+4$ and $\delta(G) \geq 1$, then $G$ has a matching $M$ of size $|M| \geq \lceil \frac{n-2}{2}
\rceil$.  \end{lemma} 

\begin{lemma}\label{matchings12} Let $G^{c}$ be a $2$-edge-coloured multigraph on $n\geq 14$ vertices coloured with $\{r,b\}$ (red and blue). 
If $rd(G^c)=2$ and $m\geq \binom{n}{2} + \binom{n-3}{2}+4$, then $G^c$ has two
matchings $M^r$ and $M^b$ of colours red and blue respectively, such that $|M^r|=\lfloor
\frac{n}{2} \rfloor$ and $|M^b| \geq \lceil \frac{n-2}{2} \rceil$.  \end{lemma}

\begin{proof} Let $E^{r}$ and $E^{b}$ denote the set of edges
coloured in red and blue respectively. Set $|E^{r}| =m^{r}$ and $|E^{b}|
=m^{b}$. Observe that, as for every vertex $x$ in $G^c$,
$rd(x)=2$, we have that $d^i(x) \geq 1$ for $i\in \left\{ r,b\right\}$. Observe
also that $m^{i}\geq \binom{n-3}{2}+4$ for $i\in \left\{ r,b\right\}$,
since this threshold is tight when the multigraph is complete on one of the
colours.

Now, if $n$ is odd, by Lemma~\ref{matching} there exist two
matchings $M^{r}$ and $M^{b}$, each one of size $\frac{n-1}{2}$, so the result follows
straightforward. Next, if $n$ is even, suppose without loss of generality that $m^{r} \geq m^{b}$.
Then $m^{r} \geq (\binom{n}{2} + \binom{n-3}{2}+4)/2 > \binom{n-2}{2}+3$. It is sufficient to show that 
$G^{r}$ has a matching of size $\lfloor \frac{n}{2} \rfloor$ because $G^{b}$ has one of size $\lceil \frac{n-2}{2} \rceil$ by
Lemma~\ref{matching}. Since $\delta(G^{r}) \geq 1$, $G^{r}$ is connected, thus, Lemma~\ref{matchingsPerfect} implies that 
$G^{r}$ has a matching of size $\lfloor\frac{n}{2} \rfloor$ as desired.
\end{proof}

\begin{lemma}\label{edges_without_cycle} Let $G^c$ be a connected
$c$-edge-coloured multigraph, $c\geq 2$. Suppose that $G^c$ contains a proper path
$P=x_1y_1x_2y_2\ldots x_py_p$, $p\geq 3$, such that each edge $x_iy_i$ is red.
If $G^c$ does not contain a proper cycle $C$ such that $V(C)=V(P)$ then there are at least $(c-1)(2p-2)$
missing edges in $G^c$.  \end{lemma} 

\begin{proof} We show that there are at least $2p-2$ missing edges in $G^c$ per colour different from red.
As there are $c-1$ such colours the total number of missing edges will be $(c-1)(2p-2)$ as claimed.
Let us consider some colour, say blue, different from red. 
The blue edge $x_1y_p$ cannot be in $G^c$ otherwise $x_1y_1\ldots x_py_px_1$ is
a proper cycle. Suppose that the blue edge $x_1x_i$ is present in $G^c$ for some $i=2,\ldots,p$.
Then the blue edge $y_{i-1}y_p$ cannot be in $G^c$ otherwise the
proper cycle $x_1x_i\ldots y_py_{i-1}\ldots x_1$ contradicts our
hypothesis. Therefore for each edge $y_{i-1}x_i$ either the blue
edge $x_1x_i$ or the blue edge $y_{i-1}y_p$ is missing. So there are
$p-1$ blue missing edges in $G^c$.
Now suppose that the blue edge $x_1y_i$ is present in $G^c$, for some
$i=2,\ldots,p-2$. Then the blue edge $x_{i+1}y_p$ cannot be together with the blue edges $x_iy_{i+1}$, $y_{i-1}x_{i+2}$ or 
$y_{i-1}y_{i+1}$, $x_{i}x_{i+2}$ in $G^c$,
otherwise the proper cycles $x_1y_ix_iy_{i+1}x_{i+1}y_p\ldots
x_{i+2}y_{i-1} \ldots x_1$ or $x_1y_ix_{i}x_{i+2}\ldots
y_px_{i+1}y_{i+1}y_{i-1} \ldots x_1$ contradict again our hypothesis. Then for each edge $y_ix_{i+1}$, at least one of the edges 
$x_{i+1}y_p$, $x_1y_i$ is missing in $G^c$ for $i=2,\ldots,p-2$. Therefore there are $p-3$ blue missing edges.

Up to now we have $2p-3$ blue missing edges. To obtain the last missing
edge observe that one of the blue edges $x_2y_p$, $x_1y_2$, $y_1x_3$ ($x_2y_p$, $x_1x_3$, $y_1y_2$) is missing in $G^c$, otherwise
we obtain the proper cycle $x_1y_2x_2y_p\ldots x_3y_1x_1$ ($x_1x_3\ldots y_px_2y_2y_1x_1$).
We remark that the blue edges $x_2y_p$ and $y_1x_3$ ($y_1y_2$) were not counted before. The edge $x_1y_2$ ($x_1x_3$) was
supposed to exist, otherwise, to obtain the last missing edge we consider the symmetric case, i.e., using the blue edge $x_1y_{p-1}$ (if it exists).

In conclusion there are $2p-2$ blue missing edges in $G^c$ as required. This completes the argument and the proof.
\end{proof}

\begin{lemma}\label{missing_edges_matching} Let $G^c$ be a connected
$c$-edge-coloured multigraph, $c\geq 2$. Let $M$ be a matching of $G^c$ in one colour, say red, of size $|M| \geq \lceil \frac{n-2}{2} \rceil$.  Let
$P=x_1y_1x_2y_2\ldots x_py_p$, $p\geq 3$, be a longest proper path compatible
with $M$. Let $f(n,c)$ denote the minimum number of missing edges in $G^c$ on colours different from red.
Then the following holds: 
\vspace*{-4mm}
\[ f(n,c) = \left\{ 
  \begin{array}{l l}
    (2n-4)(c-1) & \quad \text{if $n$ is even, $|M|=\frac{n}{2}$ and $2p < n$}\\
    (2n-6)(c-1) & \quad \text{if $n$ is odd, $|M|=\frac{n-1}{2}$ and $2p < n-1$}\\
    (2n-8)(c-1) & \quad \text{if $n$ is even, $|M|=\frac{n-2}{2}$ and $2p < n-2$}
  \end{array} \right. \]
\end{lemma} 
\begin{proof} Here we consider only the case $n$ is even, $|M|=\frac{n}{2}$ and $2p < n$, as the two other cases are similar.
Observe that, as the red matching $M$ has $\frac{n}{2}$ edges and by hypothesis $P$ uses $p$ edges of $M$, 
there are precisely $\frac{n-2p}{2}$ edges of $M$ in $G^c-P$. Let us denote these edges by
$e_i=w_iz_i$, $w_i,z_i \in G^c-P$, $i=1,\ldots,\frac{n-2p}{2}$. 

Suppose first that there is no proper $C$ cycle such that $V(C)=V(P)$.
Let blue be some colour different from red. By Lemma~\ref{edges_without_cycle} there are
$2p-2$ blue missing edges in the subgraph induced by $V(P)$.  Furthermore 
there are no blue edges between the vertices $x_1,y_p$ and the endpoints of every edge $e_i$. 
Otherwise if such an edge exists for some $i$, say $x_1w_i$, then the path $z_iw_ix_1y_1\ldots x_py_p$ contradicts
the maximality property of $P$.
Thus, there are $2(n-2p)$ blue missing edges. In adittion, for each edge 
$y_jx_{j+1}$, $j=1,\ldots,p-1$, at least two of the blue edges $y_jw_i$, $y_jz_i$, $x_{j+1}w_i$ and $x_{j+1}z_i$ are missing 
in $G^c$, otherwise if at least three among them exist, we can easily find a path longer than $P$, a contradiction.  
So, in this case there are $(n-2p)(p-1)$ blue missing edges.
Summing up we obtain  $(n-2+pn-2p^2)$ blue missing edges in $G^c$. As there are $c-1$ colours different
from red, we finally have a total of $(n-2+pn-2p^2)(c-1)$ missing edges in $G^c$. 
For $n$ and $c$ fixed, the minimum value of this function is obtained for $p=\frac{n-2}{2}$. Thus 
$f(n,c)= [n-2+\frac{n-2}{2}n-2(\frac{n-2}{2})^2](c-1)=(2n-4)(c-1)$ as required. 

Suppose next that there is a proper cycle $C$ such that $V(C)=V(P)$.
Then every edge (if any) between a vertex of $C$ and the endpoints of the edges $e_i=w_iz_i$ should be red. 
Otherwise if such a non red edge exists, say $x_jw_i$ for some $i$ and $j$, $x_i \in C$, then appropriately using the segment $x_jw_iz_i$ along 
with $C$, we may find a path longer than $P$, a contradiction.
Therefore there are at least $(2pn-4p^2)(c-1)$ missing edges in $G^c$.
Again, by minimizing the function we obtain $f(n,c)=(2n-4)(c-1)$ for $p=\frac{n-2}{2}$.
\end{proof}

\section{2-edge-coloured multigraphs}\label{2_edge_col}

In this section we study the existence of proper Hamiltonian paths in $2$-edge-coloured 
multigraphs. We present two main results. The first one involves the number of edges. 
The second one involves both the number of edges and the rainbow degree. All results are tight.

\begin{theorem}\label{s2} Let $G^c$ be a $2$-edge-coloured multigraph on $n \geq 8$ vertices coloured with $\{r,b\}$. 
If $m\geq \binom{n}{2} + \binom{n-2}{2} + 1$, then $G^c$ has a proper Hamiltonian path.  
\end{theorem} 
For the extremal example, $n \geq 8$, consider a
rainbow complete $2$-edge-coloured multigraph on $n-2$ vertices, $n$ odd. Add
two new vertices $x_1$ and $x_2$. Then add a red edge $x_1x_2$ and all red
edges between $\{x_1,x_2\}$ and the complete graph. Although the resulting
graph has $\binom{n}{2} + \binom{n-2}{2}$ edges, it has no proper Hamiltonian path, since there is no blue matching of size $(n-1)/2$. 

\begin{proof} By induction on $n$. For $n=8,9$ by a rather tedious but easy analysis the result can be shown.
Suppose now that $n\geq 10$. As $G^c$ has at least $\binom{n}{2} + \binom{n-2}{2}$ edges then $|E(\overline{G^c})| \leq 2n-4$. 
By a theorem in~\cite{Abouelaoualim2010}, if every vertex $x \in G^c$ has $d^r(x)\geq \left\lceil \frac {n+1}{2} \right\rceil $ and
$d^b(x)\geq \left\lceil \frac {n+1}{2}\right\rceil $, then $G^c$ has a proper
Hamiltonian path. Thus, we can assume that there exists a vertex $x \in G^c$ such that $d^r(x)\leq \lceil \frac {n+1}{2}\rceil -1$, otherwise 
there is nothing to prove.

Suppose first that there exist two distinct neighbours $y,z$ of $x$ such that $c(xy)=b$ and $c(xz)=r$. We then construct 
a new multigraph $G'^c$ by replacing the vertices $x,y,z$ to a new vertex $s$ such that $N^r(s)=N_{G^c-\{x,z\}}^r(y)$ and 
$N^b(s)=N_{G^c-\{x,y\}}^b(z)$. We remark that $N^r_{G^c-\{x,z\}}(y)$ and $N^b_{G^c-\{x,y\}}(z)$ cannot both be empty, otherwise 
$|E(\overline{G^c})| \geq 3n-5-\lceil \frac {n+1}{2}\rceil > 2n-4$, a contradiction.
By doing, in the worst case we remove at most $n-1$ blue and $\lceil \frac {n+1}{2}\rceil -1$ red edges from $x$, $n-3$ blue edges 
from $y$, $n-3$ red edges from $z$ and one red and one blue between $y$ and $z$. Therefore $G'^c$ has at least 
$\binom{n}{2} + \binom{n-2}{2} + 1 - (n-1) -(\lceil \frac {n+1}{2}\rceil-1) - 2(n-3) - 2  \geq \binom{n-2}{2} + \binom{n-4}{2}+1$ edges.
Thus by induction, $G'^c$ has a proper Hamiltonian path $P$. 
From this path $P$ we can easily obtain a proper Hamiltonian path in $G^c$. 

Suppose now that there does not exist two distinct neighbours $y,z$ of $x$ such that $c(xy)=b$ and $c(xz)=r$. 
Suppose first that both $y$ and $z$ exist but they are not distinct, i.e., $y=z$.
In this case, it is easy to observe that $G^c-\{x\}$ has $(n-1)(n-2)$ edges, i.e., 
it is a rainbow complete multigraph. Therefore, it contais a proper Hamiltonian path starting at $y$. This path can be easily 
extended to a proper Hamiltonian path of $G^c$ by adding one of the edges $xy$ in the appropriate colour. 
Suppose next that all edges incident to $x$ are on the same colour, say $b$. Observe that for every vertex $w \neq x$, 
there exists at least one red edge $wu$, $u \in G^c-\{x,w\}$, otherwise $|E(\overline{G^c})| \geq 2n-3 > 2n-4$, which is a contradiction. 
In the following we distinguish between to cases depending on the neighbourhood of $x$.
Assume first that $N^b(x) \leq n-2$. Consider a neighbour $y$ of $x$ and remove all its blue incident edges. Then remove 
$x$ from $G^c$ and call this multigraph $G'^c$. In $G'^c$, $y$ is monochromatic in red and $G'^c$ has at least 
$\binom{n-1}{2} + \binom{n-3}{2}+1$ edges. Thus by the inductive hypothesis, $G'^c$ 
has a proper Hamiltonian path. This path starts at $y$ since it was monochromatic. So we have a proper Hamiltonian path in $G^c$. 
Assume next that $N^b(x)=n-1$. If for some neighbour $y$ of $x$, $N^b(y)\leq n-3$, we complete the argument as before. 
Otherwise for every vertex $y$, $N^b(y)= n-2$. It follows that the underlying blue subgraph $G'^b$ of $G'^c=G^c-\{x\}$ is complete. Furthermore, $G'^c$ 
has at least $n^2 -4n+5$ edges. Now remove all the blue edges from $G'^c$. This new (red) graph has $n-1$ vertices and at least 
$\binom{n-2}{2}+1$ edges. Therefore by a theorem in~\cite{ByerSmeltzerDM2007}, it has a Hamiltonian path $P$. Now since $G'^b$ is complete, 
we can appropriately use some blue edges of $G'^b$ along with the edges of $P$ to define a proper Hamiltonian path $P'$ in $G'^c$ that always 
starts with an edge on colour red. Finally, we can join $x$ to the first vertex of $P'$ in order to obtain a proper Hamiltonian path in $G^c$.
\end{proof}

\begin{theorem}\label{rd2} Let $G^c$ be a $2$-edge-coloured multigraph on $n \geq 14$ vertices coloured with $\{r,b\}$. 
If $rd(G^c)=2$ and $m\geq \binom{n}{2} + \binom{n-3}{2}+4$, then $G^c$ has a proper Hamiltonian path.  
\end{theorem} 
For the extremal example, $n\geq 14$ odd, consider a complete blue graph, say $A$, on $n-3$ vertices. Add three
new vertices $v_1,v_2,v_3$ and join them to a same vertex $v$ in $A$ with blue
edges.  Finally, superpose the obtained graph with a complete red graph on $n$
vertices. Although the resulting $2$-edge-coloured multigraph has $\binom{n}{2} + \binom{n-3}{2}+3$ edges, it has no proper Hamiltonian path 
since one of the vertices $v_1,v_2,v_3$ cannot belong to such a path. 
\begin{proof} Let us suppose that $G^c$ does not have a proper Hamiltonian path. We will show that $\overline{G^c}$ has more
than $3n-10$ edges, i.e., $G^c$ has less than $\binom{n}{2} + \binom{n-3}{2}+4$ edges,
contradicting the hypothesis of the theorem.
We distinguish between two cases depending on the parity of $n$.

\noindent\textbf{Case A:} \emph{$n$ is even}. By Lemma~\ref{matchings12}
$G^c$ has two matchings $M^r$, $M^b$ such that $|M^r|=\frac{n}{2}$ and $|M^b|
\geq \frac{n-2}{2}$.  Take two longest proper paths, say $P=x_1y_1x_2y_2\ldots
x_py_p$ and $P'=x'_1y'_1x'_2y'_2\ldots x'_{p'}y'_{p'}$, compatibles with $M^r$
and $M^b$, respectively.  

Notice now that if $2p=n$ or $2p'=n$ then we are finished. In addition, if
$2p'< n-2$, then by Lemma~\ref{missing_edges_matching} there are
at least $2n-4$ blue missing edges and $2n-8$ red ones. This gives a total of
$4n-12 > 3n-10$ missing edges, which is a contradiction. Consequently, in what follows
we may suppose that $2p=2p'=n-2$.

Suppose first that there exists a proper cycle $C$ in $G^c$ such that
$V(C)=V(P)$. Let $e=wz$ be the red edge of $M^r-E(C)$. If there exists a
blue edge $e'$ between $w$ or $z$ and some vertex of $C$, we can easily obtain a
proper Hamiltonian path considering $e,e'$ and the rest of $C$ 
in the appropriate direction. Otherwise as the multigraph is connected, all edges $e'$
between the endpoints of $e$ and $C$ are red. Now as $rd(G^c)=2$, there must
exist a blue edge $e''$ between $w$ and $z$ and therefore we can obtain a
proper Hamiltonian path just as before but starting with $e''$ instead of $e$.

Next suppose that there exists no proper cycle $C$ in $G^c$ such that
$V(C)=V(P)$. By Lemma~\ref{missing_edges_matching} there are at least $2n-4$
blue missing edges. Consider now the path $P'$ and let $v_1,w_1$ be the two
vertices of $G^c-P'$. It is clear that if there exists a blue edge joining
$v_1$ and $w_1$, then $|M^b| = \frac{n}{2}$. Thus, 
by symmetry on the colours there are at least $2n-4$ red
missing edges. This gives a total of $4n-8 > 3n-10$ blue and red missing edges, a contradiction.
Otherwise, assume that there is no blue edge between $v_1$ and $w_1$. In this
case we will count the red missing edges assuming that we cannot extend $P'$ to a
proper Hamiltonian path. If there exists no cycle $C'$ in $G^c$ such that
$V(C')=V(P')$, then by Lemma~\ref{edges_without_cycle} there are $2p'-2=n-4$ red
missing edges.  By summing up we obtain $3n-8 > 3n-10$ missing
edges, which is a contradiction. Finally, assume that there exists a proper cycle $C'$
in $G^c$ such that $V(C')=V(P')$. Set $C=c_1c_2\ldots c_{2p'}c_{1}$ where
$c(c_ic_{i+1})=r$ for $i=1,3,\ldots, 2p'-1$. If there are three or more red edges
between $\{v_1,w_1\}$ and $\{c_i,c_{i+1}\}$, for some $i=1,3,\ldots, 2p'-1$,
then either the edges $v_1c_i$ and $w_1c_{i+1}$, or $v_1c_{i+1}$ and $w_1c_{i}$
are red. Suppose $v_1c_i$ and $w_1c_{i+1}$ are red. In this case, the path
$v_1c_ic_{i-1}\ldots c_1c_{2p'}\ldots c_{i+1}w_1$ is a proper Hamiltonian one. Otherwise, 
there are at most two red edges between $\{v_1,w_1\}$ and $\{c_i,c_{i+1}\}$,
for all $i=1,3,\ldots, 2p'-1$, then there are $2p'-2=n-4$ red
missing edges.  If we sum up we obtain a total of $3n-8 > 3n-10$
missing edges, which is a contradiction.

\noindent\textbf{Case B:} \emph{$n$ is odd}. By Lemma~\ref{matchings12} $G^c$ has two
matchings $M^r$, $M^b$ such that $|M^r|=|M^b|=\frac{n-1}{2}$. As in
Case A, we consider two longest proper paths $P$ and $P'$ compatibles
with the matchings $M^r$ and $M^b$, respectively.  Suppose first that $2p < n-1$
and $2p'<n-1$. By Lemma~\ref{missing_edges_matching} there are at least $2n-6$
blue and $2n-6$ red missing edges. We obtain a total of $4n-12 > 3n-10$
missing edges, which is a contradiction.

Suppose next $2p=2p'=n-1$ (the cases where $2p < n-1$ and $2p'=n-1$, or $2p =
n-1$ and $2p'<n-1$ are similar). In the rest of the proof,
we will consider only the path $P$ since, by symmetry, the same arguments may be applied 
for $P'$. In this case we will count the blue missing edges assuming that we cannot
extend $P$ to a proper Hamiltonian path. Now let $v$ be the unique vertex in
$G^c-P$. It is clear that if there is a proper cycle $C$ in $G^c$ such that
$V(C)=V(P)$, we can trivially obtain a proper Hamiltonian path since the multigraph
is connected. Then, as there is no proper cycle $C$ in $G^c$ such that $V(C)=V(P)$, by
Lemma~\ref{edges_without_cycle} there are $2p-2=n-3$ blue missing edges.
If there exists a blue edge between $x_1$ and $x_i$, for some $i=2,\ldots,p$,
then the blue edge $vy_{i-1}$ cannot exist in $G^c$, otherwise we would obtain the proper
Hamiltonian path $vy_{i-1}\ldots x_1x_i\ldots y_p$. We can complete the
argument in a similar way if both edges $y_py_i$ and $vx_{i+1}$,
$i=1,\ldots,p-1$ exist in $G^c$ and are on colour blue.  Note that since there
is no proper cycle $C$ in $G^c$ such that $V(C)=V(P)$, then the blue edges $x_1x_i$ and $y_py_{i-1}$, $i=2,\ldots,p$ cannot exist 
simultaneously in $G^c$. Therefore there are $p-1=\frac{n-3}{2}$ blue missing edges. If we make the
sum and multiply it by two (since the same number of red missing edges is
obtained with $P'$), we conclude that there are $3n-9 > 3n-10$ missing edges, which is a contradiction.  
This completes the argument and the proof of the theorem.
\end{proof}

\section{c-edge-coloured multigraphs, $c\geq 3$}\label{3_edge_col}

In this section we study the existence of proper Hamiltonian paths in $c$-edge-coloured 
multigraphs, for $c\geq 3$. We present three main results that involve: (1) the number of edges, (2) 
the number of edges and the connectivity of the multigraph, (3) the number of edges and the rainbow degree.
All results are tight. 

In the next lemma we present a key result that reduces the case $c\geq 4$ to $c = 3$.

\begin{lemma}\label{to3colours} Let $\ell$ be a positive integer. Let $G^c$ be a $c$-edge-coloured connected
multigraph on $n$ vertices and $m \geq c \ \ell + 1$ edges, $c\geq 4$. There
exists one colour $c_j$ such that if we colour the edges of $G^{c_j}$ with another colour
and we delete parallel edges with the same colour, then the resulting
$(c-1)$-edge-coloured multigraph $G^{c-1}$ is connected and has $m' \geq (c-1) \ell + 1$
edges. Furthermore, if $G^{c-1}$ has a proper Hamiltonian path then $G^c$ has one
too. Also, if $rd(G^c)=c$, then $rd(G^{c-1})=c-1$.
\end{lemma} 
\begin{proof} Let $c_i$ denote the colour $i$, for $i=1,\ldots,c$, and denote by $|c_i|$ the number of edges of $G^c$ with colour $i$. 
Let $c_j$ be the colour with the least number of edges.  Colour the edges on colour $c_j$ with
another colour, say $c_l$, and delete (if necessary) parallel edges with that colour. Call this multigraph $G^{c-1}$. By this, we delete at 
most $|c_j|$ edges. It is clear that this multigraph is connected since we deleted just parallel
edges. Also if $G^{c-1}$ has a proper Hamiltonian path, then this path is
also proper Hamiltonian in $G^{c}$ but perhaps with some edges on colour $c_j$ (in
the case that they have been recoloured with $c_l$). Observe also that, if
$rd(G^c)=c$ then $rd(G^{c-1})=c-1$ since only the colour $c_j$ was removed. We
will show now that $m' \geq (c-1) \ell + 1$. Now, if
$|c_j| > \ell$, then clearly $m' \geq (c-1) \ell + 1$ edges since for all
$i$, $|c_i| > \ell$.  Otherwise $|c_j| \leq \ell$. Now,
$m=\sum_{i=1}^c |c_i| \geq c \ \ell + 1$ and therefore $\sum_{i=1,i\neq j}^{c}
|c_i| \geq c \ \ell - |c_j| + 1 = (c-1)  \ell + \ell - |c_j| + 1$. This last
expression is greater than or equal to $(c-1)  \ell + 1$ since $\ell - |c_j| \geq
0$. Finally, we have that $G^{c-1}$ has $m' \geq (c-1) \ell + 1$ edges as
desired.  
\end{proof}

In view of Theorems~\ref{3coloursgeneral},\ref{3coloursconnected} and \ref{lapin} we need the following definition.

\begin{definition}\label{contractDos}
Let $G^c$ be a $3$-edge-coloured multigraph coloured with $\{r,b,g\}$.
Suppose that there exist two distinct vertices $x,y \in G^c$ such that $y$ is a neighbour of $x$ and either 
$|N(x)|=1$ or $N^r(x)=N^g(x)=\emptyset$.
First remove the vertex $x$. Then, remove all the edges (if any) in colours either $b,r$ or $b,g$, incident to $y$.
Finally rename the vertex $y$ to $s$. We call this process the \emph{contraction} of $x,y$ to $s$. 
\end{definition}

\begin{definition}\label{contract}
Let $G^c$ be a $3$-edge-coloured multigraph coloured with $\{r,b,g\}$.
Suppose that there exist three different vertices $x,y,z \in G^c$ such that $c(xy)=b$ and $c(xz)=r$.
Now the \emph{contraction} of $x,y,z$ is defined as follows: We replace the vertices $x,y,z$ by a new vertex $s$ such that 
$N^{b}(s) = N_{G^c-\{x,y\}}^{b}(z), N^{r}(s) = N_{G^c-\{x,z\}}^{r}(y)$ and 
$N^{g}(s) = N_{G^c-\{x,z\}}^{g}(y) \cap N_{G^c-\{x,y\}}^{g}(z)$.
\end{definition}

Notice that if $G'^c$ is the graph obtained from $G^c$ by any of the contractions above, then any proper Hamiltonian path in 
$G'^c$ can be easily transformed into a proper Hamiltonian one in $G^c$.

\begin{theorem}\label{3coloursgeneral} Let $G^c$ be a $c$-edge-coloured
multigraph on $n$ vertices, $n \geq 2$ and $c\geq 3$. If $m\geq
c\binom{n-1}{2}+1$, then $G^c$ has a proper Hamiltonian path.
\end{theorem} 

For the extremal case consider a rainbow complete multigraph on $n-1$ vertices with $c$ colours and add a new isolated vertex
$x$. Although the resulting multigraph has $c\binom{n-1}{2}$ edges, it 
contains no proper Hamiltonian path since it is not connected.

\begin{proof} By Lemma~\ref{to3colours} we can assume that $c=3$ and let $\{r,b,g\}$ be the set of colours. 
Assume $n \geq 6$ as cases $n \leq 5$ can be checked by exhaustive methods. 
The proof is by induction on $n$. We consider two cases depending on whether $G^c$ contains a monochromatic vertex or not.

\noindent\textbf{Case A:} \emph{There exists a monochromatic vertex $x \in G^c$}.
Assume without loss of generality that all the edges incident to $x$ are on colour $r$.
Suppose first that $d(x) \leq n-2$. Consider the multigraph $G'^c$ obtained from $G^c$ 
by contracting $x$ and one of its neighbours, say $y$, to a vertex $s$ as in Definition~\ref{contractDos} considering $r$ instead of $b$.
By this, we delete at most $3n-6$ edges. This multigraph $G'^c$ has $n-1$ vertices and at least $3\binom{n-2}{2}+1$ edges. 
Then by inductive hypothesis it has a proper Hamiltonian path. Since $s$ is monochromatic, we easily extend the path with $x$ to obtain a proper
Hamiltonian path in $G^c$.
Suppose next that $d(x)=n-1$. Then the multigraph $G^c-\{x\}$ has at least $3\binom{n-2}{2}+1$ edges and therefore by inductive hypothesis 
it has a proper Hamiltonian path $P=x_1x_2\ldots x_{n-1}$. 
Now if $c(x_1x_2)\neq r$ or $c(x_{n-2}x_{n-1})\neq r$, we are done. Otherwise, $c(x_1x_2)=c(x_{n-2}x_{n-1})=r$.
If between $x_1$ and $x_2$ there exist the three possible edges then the path $xx_1x_2\ldots x_{n-1}$ is a proper Hamiltonian one 
by appropriately choosing the edge $x_1x_2$ such that $c(x_1x_2)\neq c(x_2x_3)$ and $c(x_1x_2)\neq c(xx_1)$.
Otherwise the degree of $x_1$ in some colour different from $r$, say $b$ is at most $n-3$. Then as before, we can make the contraction with $x$ and 
$x_1$ removing the edges on colours $b$ and $r$ incident to $x_1$.

\noindent\textbf{Case B:} \emph{There is no monochromatic vertex in $G^c$}.  
Suppose first that there exists a vertex $x$ such that $|N(x)|=1$. Let $y$ be its unique neighbour. 
Now by contraction of $x$ and $y$ as in Definition~\ref{contractDos} and by deleting
edges incident to $y$ in two appropriate colours we can complete the argument.
Assume therefore that $|N(x)|\geq 2$ for all $x \in G^c$.
Moreover we may suppose that there exists a vertex $x$ such that $d(x) \leq 3n-6$. Otherwise, if for all $x\in G^c$, $d(x) \geq 3n-5$, 
then $d^i(x)\geq \left\lceil \frac{n}{2} \right\rceil \forall x\in G^c, i\in \{r,g,b\}$. Thus by a theorem in~\cite{Abouelaoualim2010},  
$G^c$ has a proper Hamiltonian cycle and so a proper Hamiltonian path.
Consider now $G^c-\{x\}$. This multigraph has at least $3\binom{n-2}{2}+1$ edges,
then by the inductive hypothesis it has a proper Hamiltonian path $P=x_1x_2\ldots x_{n-1}$.
We try to add $x$ to $P$ in order to obtain a proper Hamiltonian path in $G^c$. If $x$ is adjacent to either $x_1$ or $x_{n-1}$ in any 
appropriate colour we are done. Otherwise there are four missing edges incidet to $x$. If there are at least five edges between $x$ and some pair of 
vertices $\{x_i,x_{i+1}\}$, $i=2,\ldots,n-2$, then by choosing the appropriate edges $xx_i$ and $xx_{i+1}$, the path 
$x_1\ldots x_ixx_{i+1}\ldots x_{n-1}$ is proper Hamiltonian one in $G^c$. Otherwise 
there are at most four edges between $x$ and every pair of vertices $\{x_i,x_{i+1}\}$, for $i=2,\ldots,n-2$. Therefore there are at least 
$n-3 \geq 3$ missing edges incident to $x$. It follows that the degree of $x$ is at most $3(n-1)-4-(n-3)=2n-4 \leq 3n-10$.
Take now $y,z\in N(x)$ and suppose that $c(xy)=b$ and $c(xz)=r$. Contract $x,y,z$ as in Definition~\ref{contract}.
By this operation we remove at most $3n-10$ edges incident to $x$ and at most $3n-6$ edges incident to $y$ and $z$ in $G^c-\{x\}$. 
It follows that the obtained multigraph on $n-2$ vertices has at least $c\binom{n-1}{2}+1-(3n-10)-(3n-6) \geq c\binom{n-3}{2}+1$ edges. Therefore, 
by the inductive hypothesis it has a proper Hamiltonian path $P$. Now it is easy to obtain from $P$ a proper Hamiltonian path in $G^c$.
\end{proof}

Notice that in the above theorem there is no condition guaranteeing the
connectivity of the underlying graph. In view of Theorem~\ref{3coloursconnected} that adds this condition, we establish the following lemma.

\begin{lemma}\label{to3coloursCon}
Let $G^c$ be a $c$-edge-coloured multigraph on $n$ vertices fullfilling the conditions of Theorem~\ref{3coloursconnected} and $c \geq 4$.
Then either $G^c$ has a proper Hamiltonian path or $G^c$ contains a connected $(c-1)$-edge-coloured multigraph $G^{c-1}$ on $n$ vertices 
with at least $(c-1)\binom{n-2}{2}+n$ edges such that if $G^{c-1}$ has a proper Hamiltonian path then $G^c$ has one too.
\end{lemma}
\begin{proof} Let $c_i$ denote the colour $i$ and $E^i$ the set of edges of $G^c$ on colour $c_i$, for $i=1,\ldots,c$. 
Suppose first that there is a colour $c_j$ such that $|E^j| \leq \binom{n-2}{2}$.
Then, colour the edges on colour $c_j$ with another colour, say $c_l$, and delete (if necessary) parallel edges with the same colour. 
Call this multigraph $G^{c-1}$. Clearly $G^{c-1}$ is connected and it has at least $(c-1)\binom{n-2}{2}+n$ edges. Moreover
if $G^{c-1}$ has a proper Hamiltonian path, then it also does $G^{c}$.
Suppose next that for every colour $c_j$, $|E^j| \geq \binom{n-2}{2}+1$. If we proceed as above and we 
obtain that the multigraph $G^{c-1}$ has at least $(c-1)\binom{n-2}{2}+n$ edges, we are done. Otherwise, for each pair of colours
$c_j, c_l$ we have that $|E^j \cap E^l| \geq \binom{n-2}{2}+1$, that is, after colouring the edges on colour $c_j$ with colour $c_l$, 
there are at least $\binom{n-2}{2}+1$ parallel edges on colour $c_l$.
Now take any two colours $c_j,c_l$ and consider the uncoloured simple graph $G$ having same vertex set 
as $G^c$ and for each pair of vertices $x,y$ we add the uncoloured edge $xy$ in $G$ if and only if $xy \in E^j$ and $xy \in E^l$ in $G^c$. 
Clearly $G$ has at least $\binom{n-2}{2}+1$ edges. We distinguish between two cases depending on the connectivity of $G$.\\

Suppose first that $G$ is connected. Add a new vertex $v$ to $G$ and join it to all the vertices of $G$. 
Then $G+\{v\}$ has at least $m\geq \binom{n-1}{2} +3$ edges.
Therefore by~\cite{ByerSmeltzerDM2007}, $G+\{v\}$ is Hamiltonian-connected, that is, each pair of vertices in $G$ is joined by a Hamiltonian path.
In particular we have a Hamiltonian path $P$ that starts at $v$. Therefore if we remove $v$ from $P$ and we take its edges on
alternating colours $c_j$, $c_l$ we obtain a proper Hamiltonian path in $G^c$.

Suppose next that $G$ is disconnected. By a simple calculation on the number of edges of $G$ we can see that $G$ has two components,
say $A$ and $B$, such that either $|A|=1$ and $|B|=n-1$, or $|A|=2$ and $|B|=n-2$.\\
If $|A|=2$ and $|B|=n-2$, let $v,w$ be the vertices of $A$. By the condition on the number of edges, both $A$ and $B$ are complete.
Now, as $G^c$ is connected there exists one edge between $v$ (or $w$) and some vertex $u \in B$ on colour $c_k$.
Therefore we obtain a proper Hamiltonian path in $G^c$ starting with the edge $wv$ on colour $c_j$ (or $c_l$), then $vu$ on colour $c_k$ and 
following any Hamiltonian path in $B$ alternating the colours $c_j$, $c_l$.\\
If $|A|=1$ and $|B|=n-1$, then let $v$ be the unique vertex of $A$.
Now by~\cite{ByerSmeltzerDM2007}, $B$ has a Hamiltonian cycle unless it is isomorphic to a complete graph 
on $n-2$ vertices plus one vertex, say $w$, joint to exactly one vertex, say $u$, of the complete graph $B-\{w\}$. Now if $B$ has a Hamiltonian cycle 
$C$, then as $G^c$ is connected, there exists one edge between $v$ and some vertex in $B$ in some colour, say $c_k$. Therefore we obtain a proper 
Hamiltonian path in $G^c$ starting at $v$ taking this edge on colour $c_k$, then following $C$ alternating the colours $c_j$, $c_l$.
Alternatively, if $B$ has no Hamiltonian cycle, then $B-\{w\}$ has a Hamiltonian 
path between every pair of vertices. As $G^c$ is connected there exists one edge between $v$ and some vertex $z \in B$ on some colour $c_k$.
If $z\neq u,w$, then taking the edge $vz$ on colour $c_k$, following 
a Hamiltonian path in $B-\{w\}$ that starts at $z$ and ends at $u$ alternating the colours $c_j$, $c_l$ and taking the appropriate edge $uw$ 
we obtain a proper Hamiltonian path in $G^c$. If $z = w$, take the edge $vw$ on colour $c_k$, the edge $wu$ on colour either $c_j,c_l$
and then follow any Hamiltonian path in $B$ starting at $u$, alternating the colours $c_j$, $c_l$, we obtain a proper Hamiltonian path in $G^c$.
If none of the two above cases hold, then $v$ has only one neighbour in $B$ and $z=u$. Consider the following two cases.\\
\textbf{Case A:} \textit{The edge $vu$ exists on colour $c_k \neq c_j,c_l$}. Then, as $G^c$ has at least $m\geq c\binom{n-2}{2}+n$ edges and 
$2c<n$, $w$ has a neighbour, say $x$, in $B-\{u,w\}$ on some colour $c_s$. Then we obtain a proper Hamiltonian path in $G^c$ as follows. 
Take the edge $vu$ on colour $c_k$, continue with the edge $uw$ on colour $c_j$ or $c_l$ (depending on the colour $c_s$) and the edge $wx$ on 
colour $c_s$. Last, follow any Hamiltonian path in $B-\{u,w\}$ starting at $x$ by appropriately alternating the colours $c_j$, $c_l$.\\
\textbf{Case B:} \textit{The edge $vu$ exists only on colour $c_j$ or $c_l$, say $c_j$, but not both}. Now, by a similar argument as in case A, 
$w$ has a neighbour, say $x$, in $B-\{u,w\}$ on some colour $c_s$. 
Let $P$ be an alternating Hamiltonian path in $B-\{w\}$ from $u$ to $x$ 
such that its first edge is on colour $c_l$ and its last edge has colour different of $c_s$ 
(this is always possible because of the number of edges of $G^c$). 
Now we obtain a proper Hamiltonian path between $v$ and $w$ in $G^c$ as follows.
Add the edge $vu$ on colour $c_j$ to $P$ and complete the path with the edge $xw$ on colour $c_s$.

This completes the argument and the proof.
\end{proof}

\begin{theorem}\label{3coloursconnected} Let $G^c$ be a connected
$c$-edge-coloured multigraph on $n$ vertices, $n \geq 9$ and $3 \leq c < \frac{n}{2}$. If
$m\geq c\binom{n-2}{2}+n$, then $G^c$ has a proper Hamiltonian path.
\end{theorem} 

For the extremal example, $n \geq 9$, consider a
rainbow complete multigraph on $n-2$ vertices with $c$ colours and add two new
vertices  $x$ and $y$. Now add the edge $xy$ and all edges between $y$ and
the complete multigraph, all on the same colour. The resulting multigraph, although it has
$c\binom{n-2}{2}+n-1$ edges, it does not contain a proper Hamiltonian path
as $x$ cannot belong to such a path.

\begin{proof} 
By Lemma~\ref{to3coloursCon} we can assume that $c=3$. Let $\{r,b,g\}$ be the set of colours. The proof is by induction on $n$. For $n=9,10$ it can 
be shown by case analysis that the result holds. Now we have two cases depending on whether $G^c$ contains a monochromatic vertex or not.

\noindent\textbf{Case A:} \emph{There exists a monochromatic vertex $x \in G^c$}. 
Notice that among all neighbours of $x$ there exists at least one, say $y$, that is not monochromatic, otherwise
we would have a contradiction on the number of edges. Suppose that $c(xy)=b$.
Now we will contract $x,y$ to a new vertex $s$ as in Definition~\ref{contractDos}. Here the resulting multigraph on 
$n-1$ vertices has to be connected (as we will show later) and we need to delete at most $3n-8$ edges for the induction hypothesis to hold. 

Let us now consider $d^b(x)$. Observe that if $d^b(x)\le n-4$, we delete at most
$3n-8$ edges from $x$ and any selected neighbour $y$ of $x$ and we
are done. Further, from~\cite{Abouelaoualim2010}, if $d^i(z)\geq
\left\lceil \frac {n}{2} \right\rceil$, $\forall z\in G^c - \{x\}, i\in
\{r,g,b\}$, then $G^c-\{x\}$ has a proper Hamiltonian cycle. This would imply a
proper Hamiltonian path in $G^c$. Thus, we may assume that there exists some
vertex $w \in G^c-\{x\}$ such that $d^i(w)< \left\lceil \frac {n}{2}\right\rceil$ 
for some $i\in \{r,g,b\}$.

\noindent\textbf{Subcase A1:} \emph{$d^b(x)=n-1$}. Observe that $w\in N^b(x)$. In this case, considering $w$ instead of $y$, the contraction
process deletes $n-1$ edges from $x$, and at most $n+\frac{n}{2}-3$ from $w$, which is much less than $3n-8$ for $n>10$.
 
\noindent\textbf{Subcase A2:} \emph{$d^b(x)=n-2$}. If there is 
a vertex $y$ adjacent to $x$ such that $d_{G^c-\{x\}}^b(y)+d_{G^c-\{x\}}^r(y) \leq 2n-6$ or $d_{G^c-\{x\}}^b(y)+d_{G^c-\{x\}}^g \leq 2n-6$, 
then we just take $x$ and $y$ for the contraction process. Otherwise for all $y$ adjacent to $x$ we have 
$d_{G^c-\{x\}}^b(y)+d_{G^c-\{x\}}^r(y) \geq 2n-5$ and $d_{G^c-\{x\}}^b(y)+d_{G^c-\{x\}}^g(y) \geq 2n-5$. 
That implies $d^i(y)\geq \left\lceil \frac {n-2}{2} \right\rceil$, $\forall y\in G^c - \{x,z\}, i\in \{r,g,b\}$, where $z$ is the unique 
non-neighbour of $x$. Then by~\cite{Abouelaoualim2010}, $G^c-\{x,z\}$ has a proper Hamiltonian cycle. Finally, we can add $x$ and $z$ to the 
cycle using the fact that $x$ is adjacent to every vertex on it (as it is $z$) by the degree condition of the vertices of the cycle. 
By this we obtain a proper Hamiltonian path in $G^c$.

\noindent\textbf{Subcase A3:} \emph{$d^b(x) = n-3$}. This case is similar to the previous one but finding 
a vertex $y$ adjacent to $x$ such that $d_{G^c-\{x\}}^b(y)+d_{G^c-\{x\}}^r(y) \leq 2n-5$ or $d_{G^c-\{x\}}^b(y)+d_{G^c-\{x\}}^g(y) \leq 2n-5$.
Otherwise the multigraph $G^c-\{x\}$ is rainbow complete (except maybe for the three edges between the two non-neighbours of $x$), we 
easily find a proper Hamiltonian cycle in $G^c-\{x\}$ and then adding $x$, a proper Hamiltonian path in $G^c$.

\noindent\textbf{Case B:} \emph{There is no monochromatic vertex in $G^c$}.  
If there exists a vertex $x$ such that $|N(x)|=1$ we proceed as in case B of Theorem~\ref{3coloursgeneral}.
In what follows we assume that $|N(x)|\geq 2$ for all $x \in G^c$.
Suppose now that there exists a vertex $x$ such that $d(x) \leq 3n-8$. Otherwise, if for all $x\in G^c$, $d(x) \geq 3n-7$, then 
$m \geq \frac{n(3n-7)}{2} \geq 3\binom{n-1}{2}+1$ and by Theorem~\ref{3coloursgeneral} the result holds.
Consider now $G^c-\{x\}$. This multigraph has at least $3\binom{n-3}{2}+n-1$ edges and it is clearly connected.
Then by the inductive hypothesis it has a proper Hamiltonian path $P$.
Now we use the same argument as in Theorem~\ref{3coloursgeneral} to add $x$ to $P$.
If we cannot add it, we obtain that $d(x) \leq 3n-15$.
Finally take $y,z\in N(x)$ such that $c(xy)=b$ and $c(xz)=r$. Contract $x,y,z$ to a new vertex $s$ as in Definition~\ref{contract}. 
By this we delete at most $6n-21$ edges, that is, $3n-15$ edges incident to $x$ and $3n-6$ edges incident to $y$ and $z$ in $G^c-\{x\}$.
Since we can delete at most $6n-19$ edges to use the inductive hypothesis, the result holds.

In order to complete the proof, we will show that, either we can find two or three appropriate vertices to contract such that the obtained multigraph $G'^c$ is 
connected or $G^c$ has a proper Hamiltonian path.\\

\noindent\textit{Contraction of two vertices:} Consider the above contraction of the vertices $x,y$ to $s$ and suppose by contradiction that $G'^c$ is 
disconnected. It can be easily shown that $G'^c$ has two components with one vertex, say $z$, and $n-2$ vertices, respectively. 
Observe first that if $z = s$ then $x$ and $y$ are both monochromatic, a contradiction with the fact that $y$ was chosen not monochromatic.
Consequently $z \neq s$. 

Suppose first that $x$ is not monochromatic. In this case $x$ has $y$ as its unique neighbour.
So, there are $3(n-2)$ missing edges at $x$ and $3(n-3)$ missing edges at $z$ since $z$ is isolated in $G'^c$.
This gives us a total of $6n-15$ missing edges in $G^c$ and this is greater than $|E(\overline{G^c})|=5n-9$ which is a contradiction.

Suppose next that $x$ is monochromatic. In $G^c$ there are at least $2(n-1)$ missing edges at $x$ since it is monochromatic and $3(n-3)$ missing edges at 
$z$ since $z$ is isolated in $G'^c$. Further, there are two more missing edges between $y$ and $z$ since we have the choice of 
which colours to delete at $y$. This gives us a total of $5n-9 = |E(\overline{G^c})|$ missing edges in $G^c$. Now $z$ must be adjacent to $x$ and 
$y$ in colour $b$ otherwise we obtain $5n-8$ missing edges which is a contradiction. Therefore $z$ is also monochromatic and $d(z) = 2$. We take 
then $z$ and $y$ for the contraction (instead of $x,y$) but in this case we delete just two edges at $z$ which guarantees the connectivity of the contracted multigraph.\\

\noindent\textit{Contraction of three vertices:} Suppose by contradiction that after the contraction of
$x,y,z$ to $s$, $G'^c$ is disconnected. Then $G'^c$ has exactly two components with one vertex, say $u$, and $n-3$ vertices, respectively. 

Suppose first that $u \neq s$. In $G^c$ $u$ must have at least two different neighbours 
in two different colours among the vertices $x,y,z$. Otherwise we would be in the case where either $u$ is monochromatic or $u$ has one unique 
neighbour. Let $y'$ and $z'$ be two neighbours of $u$ among $x,y,z$ such that $c(uy') \neq c(uz')$. Now we contract
the vertices $u,y',z'$ (instead of $x,y,z$). Observe that at $u$ we delete at most six edges since $u$ has only $x,y,z$ as neighbours. 
In adittion the red edge $uy$, the blue edge $uz$ and at least one green edge among $uy, uz$ are missing. At $y'$ and $z'$ we delete $3n-6$ 
edges as usual. With this contraction we delete at most $3n$ edges and therefore the contracted multigraph has at least 
$3\binom{n-3}{2}+n-9$ edges which guarantees not only the inductive hypothesis but also the connectivity for $n\geq 10$.

Suppose next that $u=s$. Then there are no red edges between $y$ and $G^c-\{x,y,z\}$ and no blue edges between $z$ and $G^c-\{x,y,z\}$.
Now, since we are not in the previous cases, $y$ has at least two different neighbours $y'$ and $z'$ such that $c(yy') \neq c(yz')$. Then we contract 
the vertices $y,y',z'$ (instead of $x,y,z$). In the contraction process we delete at most $2(n-3)$ edges between $y$ and $G^c-\{x,y,z\}$ 
(since there are no red edges), six between $y$ and the vertices $x,z$, and $3n-6$ at $y'$ and $z'$. We obtain in total at most $5n-6$ deleted 
edges. Now, this new contracted multigraph has $n-2$ vertices and at least $3\binom{n-3}{2}-n-3$ edges. 
Clearly, if the multigraph is connected we are done. Otherwise, as before, it has two components with one vertex and $n-3$ vertices, respectively. 
We can suppose that the contracted vertex is the isolated one, otherwise we are done as above. Observe now that
the component on $n-3$ vertices has at least $3\binom{n-3}{2}-n-3$, therefore it is almost rainbow complete. It is easy to prove by induction that it 
has a proper Hamiltonian cycle. Suppose now without losing generality that $c(yy')=b$ and $c(yz')=r$.
Now, in the original multigraph if we cannot add $y,y',z'$ to the proper cycle in order to obtain a proper 
Hamiltonian path (and also using the fact that the contracted multigraph is disconnected), we obtain that there are
$n-3$ red missing edges and $n-3$ green missing ones at $y'$, $n-3$ blue and $n-3$ green at $z'$, and $n-3$ red at $y$.
We obtain a total of $5n-15$ missing edges. If we have any of the edges $r,b$ or $g$ between $y'$ and $z'$, either $y$ has no green edges
at all to $G^c-\{y,y',z'\}$ leading us to a contradiction on the number of edges, or a proper Hamiltonian path can be found. So,
these three edges are missing. Similar arguments can be used if we have the edge $yy'$ or $yz'$ in colour $g$. Therefore, two more missing edges.
Now if we have the edges $yy'$ in $r$ and $yz'$ in $b$, we can do the contraction using these colours instead of the originals.
Then, either the contracted multigraph is connected and thus we obtain a proper Hamiltonian path, or we obtain a contradiction on the 
number of edges. We can conclude that at least one between these two edges is missing obtaining a total of $5n-9=|E(\overline{G^c})|$.
That implies that $G^c-\{y,y',z'\}$ is rainbow complete and we have all of the green and blue edges between $y$ and $G^c-\{y,y',z'\}$, 
all of the blue between $y'$ and $G^c-\{y,y',z'\}$, and all of the red between $z'$ and $G^c-\{y,y',z'\}$. In this last case, 
it is easy to obtain a proper Hamiltonian path in $G^c$.
\end{proof}

In view of Theorem~\ref{lapin} we prove the following lemma.

\begin{lemma}\label{lemmalapin} Let $G^c$ be a $c$-edge-coloured multigraph on $n$
vertices fullfilling the conditions of Theorem~\ref{lapin}. Then either $G^c$ has a proper Hamiltonian path or 
there exists a vertex $x \in G^c$ such that $d(x) \leq 2n-6$.
\end{lemma} 

\begin{proof}
Let $E^i$ be the set of edges of colour $i$, $i \in \{r,g,b\}$, and suppose without loss of generality that $|E^b|\geq |E^r|,|E^g|$. 
Then, as the subgraph $G^b$ has minimum degree one and $|E^b| \geq \binom{n-2}{2}+3$, it can be easily checked that it is connected. 
Thus by Lemma~\ref{matchingsPerfect} there is a matching $M^b$ such that $|M^b|=\frac{n}{2}$ for $n$ even and $|M^b|=\frac{n-1}{2}$
for $n$ odd. Let $P=x_1y_1x_2y_2 \ldots x_py_p$ be the longest proper path compatible with $M^b$. 

Suppose first that $n$ is odd. By Lemma~\ref{missing_edges_matching}, if there is a proper cycle $C$ such that $V(C)=V(P)$, then $|P|\geq n-5$. Else, if such a cycle
does not exist then $|P|\geq n-7$. Otherwise in both cases we obtain a contradiction on the number of edges.
Let us consider here the case $|P|=n-1$ (the other cases are easier to handle, refer to~\cite{leathesis} for more details).
Now observe that if there is a proper cycle $C$ such that $V(C)=V(P)$, then the result easily follows as the unique vertex of $G^c-C$ can be
appropriately joint to $C$ in order to obtain a proper Hamiltonian path.
Assume therefore that there is no proper cycle $C$ such that $V(C)=V(P)$.
Let $x$ be the unique vertex of $G^c-P$.
Clearly we cannot have either the edge $xx_1$ on colours $r$ or $g$, or the edge $xy_p$ on colours $r$ or $g$, otherwise
we easily obtain a proper Hamiltonian path in $G^c$.
Now, if there are at least three edges on colours $r,g$ between $x$ and some pair of 
vertices $\{y_i,x_{i+1}\}$, $i=2,\ldots,p-1$, then by choosing the appropriate edges $xy_i$ and $xx_{i+1}$, the path 
$x_1\ldots y_ixx_{i+1}\ldots y_p$ is a proper Hamiltonian one in $G^c$. Otherwise 
there are at most two edges on colours $r,g$ between $x$ and every pair of vertices $\{y_i,x_{i+1}\}$, for $i=2,\ldots,p-1$. 
Therefore $d^{r,g}(x) \leq n-3$ and clearly $d(x) \leq 2n-4$ as $d^b(x) \leq n-1$. In addition, if we have two more missing edges incident to $x$
we would obtain that $d(x) \leq 2n-6$ as claimed. Now, we can assume the worst case, that is, for each edge $y_ix_{i+1}$ in the path, 
$i=2,\ldots,p-1$, we have both edges $xy_i$,$xx_{i+1}$ on the same colour of $y_ix_{i+1}$ (that is, $r$ or $g$). Otherwise, if we suppose 
without losing generality that $c(xy_i)=r$ and $c(y_ix_{i+1})=g$ then we cannot have the blue edge $xx_i$ (or we would obtain the proper
Hamiltonian path $x_1\ldots x_ixy_i\ldots y_p$). Therefore, there would be one more missing edge at $x$.
Consider now $x_1$. Suppose that we have any edge $x_1y_i$ on colour $r$ or $g$ that is different of the colour of $y_ix_{i+1}$, for $i=1,\ldots,p-1$.
Taking the blue edge $xx_i$ we obtain the proper Hamiltonian path in $G^c$, 
$xx_i\ldots x_1y_ix_{i+1}\ldots y_p$. Otherwise we obtain at least $p-1 = \frac{n-3}{2}$ missing edges $x_1y_i$ on colours $r$ or $g$.
Suppose that we have any edge $x_1x_i$ on at least one colour $r$ or $g$, for $i=2,\ldots,p$. Therefore taking the edge $xy_{i-1}$ on colour
$r$ or $g$ (one of both is supposed to exist) we obtain the proper Hamiltonian path $xy_{i-1}\ldots x_1x_iy_i\ldots y_p$. 
Otherwise the edges $x_1x_i$ on colours $r$ and $g$ are missing for all $i=2,\ldots,p$, that is, $2(p-1)=n-3$ additional missing edges at $x_1$.
Finally, summing up and considering that we cannot have the edge $x_1y_p$ on colours $r$ or $g$ (or $P$ would also be a proper cycle), we obtain
that $d(x) \leq 3(n-1) - \frac{n-3}{2} - (n-3) - 2 \leq 2n-6$ as claimed.

Suppose next that $n$ is even. If there is a proper cycle $C$ such that $V(C)=V(P)$, then by Lemma~\ref{missing_edges_matching} $|P| \geq n-2$.
This case is easy since either $P$ is a proper Hamiltonian path or we can connect the unique edge of $M^b - E(P)$ to $C$ 
in order to obtain a proper Hamiltonian path. Assume therefore that there is no proper cycle $C$ such that $V(C)=V(P)$. It follows by
Lemma~\ref{missing_edges_matching} that $|P|\geq n-4$ otherwise we obtain a contradiction on the number of edges. 
Let us consider just the case $|P|=n-2$ ($|P|=n-4$ is easier, refer to~\cite{leathesis} for full details).
Let $e=xy$ be the edge of $M^b - E(P)$. Now by similar arguments as in the odd case above, we can prove that, either the edge $e$ can be added to $P$ in 
order to obtain a proper Hamiltonian path in $G^c$, or one of the vertices $x,y,x_1,y_p$ has degree at most $2n-6$ as required.
\end{proof}

\begin{theorem}\label{lapin} Let $G^c$ be a $c$-edge-coloured multigraph on $n$
vertices, $n \geq 11$ and $c\geq 3$. 
If $rd(G^c)=c$ and $m\geq c\binom{n-2}{2}+2c+1$, then $G^c$ has a
proper Hamiltonian path.  
\end{theorem} 

For the extremal example, $n \geq 11$, consider a
rainbow complete multigraph, say $A$, on $n-2$ vertices. Add two new vertices
$v_1,v_2$ and join them to a vertex $v$ of $A$  with all possible colours.
The resulting $c$-edge-coloured multigraph has $c\binom{n-2}{2}+2c$ edges
and clearly has no proper Hamiltonian path. 

\begin{proof} By Lemma~\ref{to3colours}
it is enough to prove the theorem for $c=3$. Let $\{r,b,g\}$ be the set of colours. As $m\geq 3\binom{n-2}{2}+7$ then $|E(\overline{G^c})|\leq 6n-16$. The proof will be
done either by construction of a proper Hamiltonian path or using Theorem~\ref{3coloursconnected}.
We will do this by contracting two or three vertices depending on if there exists a vertex $x$ in $G^c$ such that
$|N(x)|=1$ or not.

If there exists a vertex $x\in G^c$ such that $|N(x)|=1$ we contract $x$ and its unique neighbour $y$ to a new vertex
$s$ as in Definition~\ref{contractDos}. By this we delete at most $2n-1$ edges and the resulting multigraph is still connected. 
Thus the conclusion follows from Theorem~\ref{3coloursconnected}.

Suppose next that there is no vertex $x \in G^c$ such that $|N(x)|=1$. It follows that for any vertex $x$ there are two distinct neighbours
$y$,$z$ in $G^c$ such that $c(xy)=b$ and $c(xz)=r$. 
Now by Lemma~\ref{lemmalapin} consider a vertex $x$ such that $d(x)\leq 2n-6$.
Then contract $x,y,z$ to a new vertex $s$ as in Definition~\ref{contract}.
Let $G'^c$ be the resulting multigraph. In this case, as we delete at most $5n-12$ ($=2n-6 + 3(n-3) + 3$) edges, 
it is enough to prove that $G'^c$ is connected to apply Theorem~\ref{3coloursconnected}. 

Suppose therefore by contradiction that $G'^c$ is disconnected. Then it has exactly two components with one vertex, say $u$, and $n-3$ vertices, 
respectively, otherwise we arrive to a contradiction on the number of edges. 

Assume first that $u\neq s$. Then, as in the equivalent case of Theorem~\ref{3coloursconnected}, instead of $x,y,z$, we may find three other vertices 
$u,y',z'$ to contract to a vertex $s'$ just deleting $3n$ edges. This new obtained multigraph has at least 
$3\binom{n-3}{2}-2$ edges. Then, if it is connected we are done, otherwise there is a component with one vertex, say $u'$, and 
another one on $n-3$ vertices with at least $3\binom{n-3}{2}-2$ edges, i.e., almost rainbow complete. Therefore, the biggest component
contains a proper Hamiltonian cycle and then we can easily add either the isolated vertex $u'$ (if $u'\neq s'$) or the three 
$u,y',z'$ (if $u'= s'$) vertices to the cycle to obtain a proper Hamiltonian path in $G^c$.

Assume next $u=s$. If $d(x) \leq n+1$, then the contraction process deletes $4n-5$ edges instead of $5n-12$. Furthermore as $G'^c$ 
is disconnected by hypothesis, its component on $n-3$ vertices has at least $3\binom{n-3}{2}-n+3$ edges. 
As in Theorem~\ref{3coloursconnected}, this component is almost rainbow complete and thus it contains a proper Hamiltonian cycle $C$. 
This allows us to easily add $x,y,z$ to $C$ in order to obtain a proper Hamiltonian path in $G^c$. 
In the sequel, we may suppose that $d(x) \geq n+2$. Then $x$ has two different neighbours $y'$ and $z'$ with parallel edges. Consider the next 
two cases:\\
Assume first that the parallel edges are on the same two colours, say $c(xy') = c(xz')=\{b,r\}$ (cases with other two colours are similar). 
Here we can consider two possible contractions: 1) $x,y',z'$ with $c(xy')=b$, $c(xz')=r$ and 2) $x,y',z'$ with $c(xy') = r$, $c(xz')=b$.
Now, suppose that in both contractions the multigraph is disconnected and the contracted vertex is always the isolated one, otherwise we are finished. 
We can observe that $G^c$ has $n+3$ missing edges incident to $x$ (since $d(x) \leq 2n-6$), $n-3$ green edges and $4(n-3)$ blue and red edges 
incident to $y'$ and $z'$ (since in both contractions the multigraph is disconnected). By this we obtain a total of $6n-12 > 6n-16$ missing edges, 
which is a contradiction.\\ 
Assume next that the parallel edges are not on the same two colours, that is, $c(xy')=\{b,r\}$ and $c(xz')=\{b,g\}$ 
(cases with other combinations are similar). Now since we are not in the previous case, we do not have either the green edge $xy'$ or the red one 
$xz'$. Try any of the three possible contractions: 1) $x,y',z'$ with $c(xy')=b$, $c(xz')=g$, 2) $x,y',z'$ with $c(xy')=r$, $c(xz')=g$ and 
3) $x,y',z'$ with $c(xy') = r$, $c(xz')=b$.
Then, after each of these contractions the multigraph is still disconnected and the contracted vertex is always the isolated one. We can observe 
that there can exist just the red edges between $y'$ and $G^c-\{x,y',z'\}$ and the green edges between $z'$ and $G^c-\{x,y',z'\}$. 
Now as $rd(G^c)=3$ there must exist the green edge $y'z'$ and the red edge $y'z'$. Since we are not in the previous case, the blue edge $y'z'$ is 
not present. We find us in the situation that 
$c(xy')=\{b,r\}$, $c(xz')=\{b,g\}$ and $c(y'z')=\{r,g\}$. Now, we have nine different contractions to try, three for each triplet $x,y',z'$, 
$y',x,z'$ and $z',x,y'$. If in all of them we are in this same situation (the contracted multigraph is disconnected and the isolated vertex is the 
contracted one) we can conclude that in $G^c$ there can exist just the blue edges between $x$ and $G^c-\{x,y',z'\}$, 
the red edges between $y'$ and $G^c-\{x,y',z'\}$, and the green edges between $z'$ and $G^c-\{x,y',z'\}$. This gives a total of
$6(n-3)$ missing edges in $G^c$. Finally, adding the three missing edges $xy'$ in green, $xz'$ in red and $y'z'$ in blue, 
we obtain $6n-15$ missing edges which is a contradiction.
\end{proof}

\noindent\textbf{Acknowledgments:} The authors would like to thank N. Narayanan for his valuable corrections and comments.

\bibliography{biblio}

\end{document}